 \documentclass[10pt, conference]{IEEEtran}
\usepackage[utf8]{inputenc}
\usepackage{amsmath}
\usepackage{mathtools}
\usepackage{amsfonts}
\usepackage{amssymb}
\usepackage{amsthm}
\usepackage{subfigure}
\usepackage{cite}
\usepackage{calc}
\usepackage{color}
\usepackage{epsfig}
\usepackage{setspace}
\usepackage{pstricks}
\usepackage{cancel}
\usepackage{multirow}

\DeclarePairedDelimiter{\ceil}{\lceil}{\rceil}
\DeclarePairedDelimiter{\floor}{\lfloor}{\rfloor}

\newtheorem{defn}{Definition}
\newtheorem{thm}{{\cal T}heorem}[section]
\newtheorem{cor}[thm]{Corollary}
\newtheorem{prop}{Proposition}
\newtheorem{lem}[thm]{Lemma}
\newtheorem{conj}[thm]{Conjecture}
\newtheorem{constr}[thm]{Construction}
\newtheorem{note}{Remark}

\newtheorem{example}{Example}
\newcommand{\bit}{\begin{itemize}}
\newcommand{\eit}{\end{itemize}}
\newcommand{\bcor}{\begin{cor}}
\newcommand{\ecor}{\end{cor}}
\newcommand{\beq}{\begin{equation}}
\newcommand{\eeq}{\end{equation}}
\newcommand{\beqn}{\begin{equation}}
\newcommand{\eeqn}{\end{equation}}
\newcommand{\bea}{\begin{eqnarray}}
\newcommand{\eea}{\end{eqnarray}}
\newcommand{\bean}{\begin{eqnarray*}}
\newcommand{\eean}{\end{eqnarray*}}
\newcommand{\ben}{\begin{enumerate}}
\newcommand{\een}{\end{enumerate}}
\newcommand{\bdefn}{\begin{defn}}
\newcommand{\edefn}{\end{defn}}
\newcommand{\bnote}{\begin{note}}
\newcommand{\enote}{\end{note}}
\newcommand{\bprop}{\begin{prop}}
\newcommand{\eprop}{\end{prop}}
\newcommand{\blem}{\begin{lem}}
\newcommand{\elem}{\end{lem}}
\newcommand{\bthm}{\begin{thm}}
\newcommand{\ethm}{\end{thm}}
\newcommand{\bconj}{\begin{conj}}
\newcommand{\econj}{\end{conj}}
\newcommand{\bconstr}{\begin{constr}}
\newcommand{\econstr}{\end{constr}}
\newcommand{\bpf}{\begin{proof}}
\newcommand{\epf}{\end{proof}}

\title{Maximally Recoverable Codes with Hierarchical Locality \vspace{-1ex}}

\author{
	\IEEEauthorblockN{Aaditya M. Nair, V. Lalitha
		\thanks{Aaditya and Dr. Lalitha are with the Signal Processing \& Communications Research Center, International Institute of Information Technology Hyderabad, India, email:aaditya.mnair@research.iiit.ac.in, lalitha.v@iiit.ac.in.}
	}
}
\IEEEoverridecommandlockouts
\begin{document}

\maketitle

 \begin{abstract}
 Maximally recoverable codes are a class of codes which recover from all potentially recoverable erasure patterns given the locality constraints of the code. In earlier works, these codes have been studied in the context of codes with locality. The notion of locality has been extended to hierarchical locality, which allows for locality to gradually increase in levels with the increase in the number of erasures. We consider the locality constraints imposed by codes with two-level hierarchical locality and define maximally recoverable codes with data-local and local hierarchical locality. We derive certain properties related to their punctured codes and minimum distance. We give a procedure to construct hierarchical data-local MRCs from hierarchical local MRCs. We provide a construction of hierarchical local MRCs for all parameters. For the case of one global parity, we provide a different construction of hierarchical local MRC over a lower field size.

 \end{abstract}
\vspace{-2ex}
\section{Introduction}

With application to distributed storage systems, the notion of locality of a code was introduced in \cite{gopalan2012locality}, which enables efficient node repair in case of single node failures (node failures modelled as erasures)  by contacting fewer nodes than the conventional erasure codes based on maximum distance separable (MDS) codes. An extension to handle multiple erasures has been studied in \cite{locality}.
A code symbol is said to have $(r, \delta)$ locality if there exists a punctured code $\mathcal{C}_i$ such that $c_i \in Supp(\mathcal{C}_i)$ and the following conditions hold, 
%1) $dim(\mathcal{C}_i) \leq r$ and, 
1) $|Supp(C_i)| \leq r+ \delta -1$ and,
2)  $d_{min}(\mathcal{C}_i) \geq \delta$

% CAUTION: Check this
An $[n,k,d_{min}]$ code is said to have $(r,\delta)$ information locality, if $k$ data symbols have $(r,\delta)$ locality and it is said to have all-symbol locality if all the $n$ code symbols have $(r,\delta)$ locality. An upper bound on the minimum distance of a code with $(r,\delta)$ information locality is given by
\begin{equation} \label{eq:dmin_rdelta}
d_{min} \leq n - k + 1 - \left ( \left \lceil \frac{k}{r} \right \rceil -1 \right ) (\delta-1).
\end{equation}

\subsection{Maximally Recoverable Codes with Locality}

Maximally recoverable codes (MRC) are a class of codes which recover from all information theoretically recoverable erasure patterns given the locality constraints of the code. Maximally recoverable codes with locality have been defined for the case of $\delta =2$ in \cite{mrc}. We extend the definitions here for the general $\delta$.

\begin{defn}[Data Local Maximally Recoverable Code]\label{defn:data_local}
Let $C$ be a systematic $[n,k,d_{min}]$ code. We say that $\mathcal{C}$ is an $[k, r, h, \delta]$ data-local maximally recoverable code if the following conditions are satisfied
\bit
    \item $r | k$ and $n=k+\frac{k}{r} \delta+h$.
    \item Data symbols are partitioned into $\frac{k}{r}$ groups of size $r$. For each such group, there are $\delta$ local parity symbols.
    \item The remaining $h$ global parity symbols may depend on all $k$ symbols.
    \item For any set $E \subseteq [n]$ where $E$ is obtained by picking $\delta$ coordinates from each $\frac{k}{r}$ local groups, restricting $\mathcal{C}$ to coordinates in $[n]-E$ yields a $[k+h, k]$ MDS code. 
\eit
\end{defn}
$[k, r, h, \delta]$ data-local MRC is optimum with respect to minimum distance bound in \eqref{eq:dmin_rdelta}. The minimum distance of a $[k, r, h, \delta]$ data-local MRC is given by $d_{min} = h+\delta+1.$

\begin{defn}[Local Maximally Recoverable Code]\label{defn:local}
Let $C$ be a systematic $[n,k,d_{min}]$ code. We say that $\mathcal{C}$ is an $[k, r, h, \delta]$ local maximally recoverable code if the following conditions are satisfied
\bit
    \item $r | (k+h)$ and $n=k+\frac{k+h}{r}\delta+h$.
    \item There are $k$ data symbols and $h$ global parity symbols where each global parity may depend on all data symbols.  
    \item These $k+h$ symbols are partitioned into $\frac{k+h}{r}$ groups of size $r$. For each group there are $\delta$ local parity symbols.
    \item For any set $E \subseteq [n]$ where $E$ is obtained by picking $\delta$ coordinates from each $\frac{k+h}{r}$ local groups, restricting $\mathcal{C}$ to coordinates in $[n]-E$ yields a $[k+h, k]$ MDS code.
\eit
\end{defn}
$[k, r, h, \delta]$ local MRC is optimum with respect to minimum distance bound in \eqref{eq:dmin_rdelta}. The minimum distance of a $[k, r, h, \delta]$ local MRC is given by  
\begin{equation} \label{eq:dminLMRC}
d_{min} = h+ \delta + 1 + \floor[\Big]{\frac{h}{r}}\delta.
\end{equation}

Maximally recoverable codes with locality for the case of general $\delta$ are known in literature as Partial-MDS codes (PMDS) codes. MRCs have been studied in the context of distributed storage systems and PMDS codes in the context of solid state drives (SSD) \cite{pmds}. Constructions of PMDS codes with two and three global parities have been discussed in \cite{blaum2016construction, chen2015sector}. A general construction of PMDS codes based on linearized polynomials has been provided in \cite{calis2017general}. An improved construction of PMDS codes for all parameters over small field sizes ($\mathcal{O}(\textit{max}\{\frac{k+h}{r},(r+\delta)^{\delta+h}\}^h)$) has been presented in \cite{small}. Constructions of MRCs with field size $\mathcal{O} ((\frac{k+h}{r})^r)$ have been presented in \cite{martinez2018universal}. Construction of MRCs ($\delta=2$) over small field sizes have been investigated in \cite{hu2016new, guruswami2018constructions}.

\subsection{Codes with Hierarchical Locality}
The concept of \emph{locality} has been extended to hierarchical locality in \cite{hlocality}. In the case of $(r, \delta)$ locality, if there are more than $\delta$ erasures, then the code offers no locality. In the case of codes with hierarchical locality, the locality constraints are such that with the increase in the number of erasures, the locality increases in steps. The following is the definition of code with two-level hierarchical locality.

\begin{defn} \label{defn:hier_local}
An $[n, k, d_{min}]$ linear code $\mathcal{C}$ is a code with \emph{hierarchical locality} having parameters $[(r_1, \delta_1), (r_2, \delta_2)]$  if for every symbol $c_i$, $1 \leq i \leq n$, there exists a punctured code $\mathcal{C}_i$ such that $c_i \in Supp(C_i)$ and the following conditions hold, 
    1) $|Supp(\mathcal{C}_i)| \leq r_1+\delta-1 $
    2) $d_{min}(\mathcal{C}_i) \geq \delta_1$ and
    3) $\mathcal{C}_i$ is a code with $(r_2, \delta_2)$ locality.
\end{defn}

An upper bound on the minimum distance of a code with two-level hierarchical locality is given by 
\begin{equation} \label{eqn:min_dist_bound}
    d \leq n - k + 1 -(\ceil[\Big]{\frac{k}{r_2}}-1)(\delta_2 - 1) - (\ceil[\Big]{\frac{k}{r_1}}-1)(\delta_1 - \delta_2).
\end{equation}

\subsection{Our Contributions}

In this work, we consider the locality constraints imposed by codes with two-level hierarchical locality and define maximally recoverable codes with data-local and local hierarchical locality. We prove that certain punctured codes of these codes are data-local/local MRCs. We derive the minimum distance of hierarchical data-local MRCs. We give a procedure to construct hierarchical data-local MRCs from hierarchical local MRCs. We provide a construction of hierarchical local MRCs for all parameters. For the case of one global parity, we provide a different construction of hierarchical local MRC over a lower field size.

\subsection{Notation}
For any integer $n$, $[n] = \{1, 2, 3 \ldots, n\}$. For any $E \subseteq [n]$, $\bar{E} = [n]-E$. For any $[n,k]$ code, and any $E \subseteq [n]$, $\mathcal{C}|_E$ refers to the punctured code obtained by restricting $\mathcal{C}$ to the coordinates in $E$. This results in an $[n-|E|, k']$ code where $k' \leq k$. For any $m\times n$ matrix $H$ and $E \subseteq [n]$, $H|_E$ is the $m \times |E|$ matrix formed by restricting $H$ to columns indexed by $E$.
In several definitions to follow, we implicitly assume certain divisibility conditions which will be clear from the context.

\vspace{-1ex}
\section{Maximally Recoverable Codes with Hierarchical Locality} \label{sec:HLMRC}

In this section, we define hierarchical data-local and local MRCs and illustrate the definitions through an example. We describe these codes via their parity check matrices instead of generator matrices (data local and local MRCs were defined by their generator matrices).

\begin{defn}[Hierarchical Data Local Code] \label{defn:HDLC}
We define a $ [k, r_1, r_2, h_1, h_2, \delta] $  hierarchical data local (HDL) code of length  $n=k+h_1+\frac{k}{r_1}(h_2+\frac{r_1}{r_2}\delta)$ as follows:
\begin{itemize}
\item The code symbols $c_1, \ldots, c_n$ satisfy $h_1$ global parities given by $\sum_{j = 1}^n u_j^{(\ell)} c_j = 0, \ \ 1 \leq  \ell \leq h_1$.
\item The first $n-h_1$ code symbols are partitioned into $t_1 = \frac{k}{r_1}$ groups $A_i, 1 \leq i \leq t_1$ such that $|A_i| =  r_1 + h_2 + \frac{r_1}{r_2} \delta = n_1$. The code symbols in the $i^\text{th}$ group, $1 \leq i \leq t_1$ satisfy the following $h_2$  mid-level parities $\sum_{j = 1}^{n_1} v_{i,j}^{(\ell)} c_{(i-1)n_1+j} = 0, \ \ 1 \leq  \ell \leq h_2$.
\item The first $n_1-h_2$ code symbols of the $i^\text{th}$ group, $1 \leq i \leq t_1$ are partitioned into $t_2 = \frac{r_1}{r_2}$ groups $B_{i,s}, 1 \leq i \leq t_1, 1 \leq s \leq t_2$ such that $|B_{i,s}| =  r_2 + \delta = n_2$. The code symbols in the $(i,s)^\text{th}$ group, $1 \leq i \leq t_1, 1 \leq s \leq t_2$ satisfy the following $\delta$  local parities $\sum_{j = 1}^{n_2} w_{i,s,j}^{(\ell)} c_{(i-1)n_1+(s-1)n_2+j} = 0, \ \ 1 \leq  \ell \leq \delta$.
\end{itemize}
\end{defn}

\begin{defn}[Hierarchical Data Local MRC] \label{defn:HDLMRC}
Let $\mathcal{C}$ be a $[k, r_1, r_2, h_1, h_2, \delta]$ HDL code. Then $C$ is maximally recoverable if for any set $E \subset [n]$ such that $|E| = k+h_1$, $|E \bigcap B_{i, s}| \leq r_2$ $\forall \ i, s$ and $|E \bigcap A_i|  = r_1$ $\forall \ i$,
%\begin{enumerate}
%    \item $E \bigcap B_{i, s} \ge r_2$ $\forall \ i, s$,
%    \item $E \bigcap A_i  = r_1$ $\forall \ i$,
%\end{enumerate}
the punctured code $\mathcal{C}|_E$ is a $[k+h_1,k,h_1+1]$ MDS code. 
\end{defn}
%Similar to Lemma 4 in \cite{mrc}, we can show that the set of erasure patterns given above are the only erasure patterns which are information theoretically recoverable.

\begin{defn}[Hierarchical Local Code] \label{defn:HLC}
We define a $ [k, r_1, r_2, h_1, h_2, \delta] $ hierarchical local (HL) code of length  $n=k+h_1+\frac{k+h_1}{r_1}(h_2+\frac{r_1+h_2}{r_2}\delta)$ as follows:
\begin{itemize}
\item The code symbols $c_1, \ldots, c_n$ satisfy $h_1$ global parities given by $\sum_{j = 1}^n u_j^{(\ell)} c_j = 0, \ \ 1 \leq  \ell \leq h_1$.
\item The $n$ code symbols are partitioned into $t_1 = \frac{k+h_1}{r_1}$ groups $A_i, 1 \leq i \leq t_1$ such that $|A_i| = r_1 + h_2 + \frac{r_1+h_2}{r_2}\delta = n_1$. The code symbols in the $i^\text{th}$ group, $1 \leq i \leq t_1$ satisfy the following $h_2$  mid-level parities $\sum_{j = 1}^{n_1} v_{i,j}^{(\ell)} c_{(i-1)n_1+j} = 0, \ \ 1 \leq  \ell \leq h_2$.
\item The $n_1$ code symbols of the $i^\text{th}$ group, $1 \leq i \leq t_1$ are partitioned into $t_2 = \frac{r_1+h_2}{r_2}$ groups $B_{i,s}, 1 \leq i \leq t_1, 1 \leq s \leq t_2$ such that $|B_{i,s}| =  r_2 + \delta = n_2$. The code symbols in the $(i,s)^\text{th}$ group, $1 \leq i \leq t_1, 1 \leq s \leq t_2$ satisfy the following $\delta$  local parities $\sum_{j = 1}^{n_2} w_{i,s,j}^{(\ell)} c_{(i-1)n_1+(s-1)n_2+j} = 0, \ \ 1 \leq  \ell \leq \delta$.
\end{itemize}
\end{defn}

\begin{defn}[Hierarchical Local MRC] \label{defn:HLMRC}
Same as Definiton \ref{defn:HDLMRC}.
%Let $\mathcal{C}$ be a $[k, r_1, r_2, h_1, h_2, \delta]$ HL code. Then $C$ is maximally recoverable if for any set $E \subset [n]$ such that $|E| = k+h_1$ and
%\begin{enumerate}
%    \item $E \bigcap B_{i, s} \leq r_2$ $\forall \ i, s$,
%    \item $E \bigcap A_i  = r_1$ $\forall \ i$,
%\end{enumerate}
%the punctured code $\mathcal{C}|_E$ is a $[k+h_1,k,h_1+1]$ MDS code. 
\end{defn}

In an independent parallel work \cite{martinez2018universal}, a class of MRCs known as multi-layer MRCs have been introduced. We would like to note that hierarchical local MRCs (given in Definition \ref{defn:HLMRC}) form a subclass of these multi-layer MRCs.

\begin{example}
We demonstrate the structure of the parity check matrix for an $[k=5, r_1=3, r_2=2, h_1=1, h_2=1 , \delta=2]$ HL code.
The length of the code is $n=k+h_1+\frac{k+h_1}{r_1}(h_2+\frac{r_1+h_2}{r_2}\delta) = 16$. The parity check matrix of the code is given below:

%Let the codeword  be $ (c_1, c_2, \ldots, c_{16})$.
%%Also let $\Gamma = \{\alpha_{1,1}, \alpha_{1, 2}, \ldots, \alpha_{2, 4}\}, \; \alpha_{i,j} \in \mathbb{F}_{q^M}$. $\beta$ is a primitive element of $\mathbb{F}_q$
%All symbols satisfy one global parity constraint $\sum_{i=1}^{16} u_j^{(1)} \cdot c_{j} = 0$.
%Partition the code symbols into $t_1 = \frac{k+h_1}{r_1} = 2$ groups of size $n_1 = 8$, denoted by 
%$A_1 = (1, \ldots, 8) $ and $A_2 = (9, \ldots, 16) $. 
%Symbols corresponding to $A_1$ satisfy mid-level parity constraint $\sum_{j=1}^{8} v_{1, j}^{(1)} \cdot c_{j} = 0$. Similar mid-level parity constraint is satisfied by symbols corresponding to $A_2$.
%$A_1$ is then partitioned into $t_2 = \frac{r_1+h_2}{r_2} = 2$ of size $n_2 = r_2+\delta = 4$. Therefore, $B_{1,1} = \{1, \ldots, 4\}$ and $B_{1,2} = \{5,\ldots, 8 \}$. Since $\delta=2$, symbols corresponding to $B_{1,1}$ satisfies two local parity constraints $\sum_{i=1}^{4} w_{1,1,j}^{(1)} \cdot c_i = 0$ and $\sum_{i=1}^{4} w_{1,1,j}^{(2)} \cdot c_i = 0$. Similar local parity constraints are satisfied by symbols corresponding to $B_{1,2}, B_{2,1}, B_{2,2}$.
\[
H=\begin{bmatrix}
	\begin{matrix}
	 	\begin{array}{r|r}
	 		\begin{matrix} 
	 		\begin{matrix} 
	 		M_{1,1} & \\
	 		& M_{1,2} \\
	 		\end{matrix} \\
	 		\hline
	 		N_1 \\
	 		\end{matrix} & \\
	 		\hline
	 		& \begin{matrix} 
	 		\begin{matrix} 
	 		M_{2,1} & \\
	 		& M_{2,2} \\
	 		\end{matrix} \\
	 		\hline
	 		N_2 \\
	 		\end{matrix} \\
	 	\end{array} \\
	 	\hline
	 	P \\
	\end{matrix}
\end{bmatrix}
\]

%\[
%	H = \begin{bmatrix}
%	\begin{matrix}
%		\begin{matrix}
%		M_{1,1} \\
%		& M_{1,2} \\
%		\end{matrix} \\
%	
%		N_1 \\
%	
%		& \begin{matrix}
%			M_{2,1} \\
%			& M_{2,2} \\
%		  \end{matrix} \\
%	
%		& N_2 \\
%	\end{matrix}\\
%	
%	O \\
%	\end{bmatrix}
%\]

%where, //CHANGE

\[
	M_{i,j} = \begin{bmatrix}
	w_{i,j,1}^{(1)} & w_{i,j,2}^{(1)} & w_{i,j,3}^{(1)} & w_{i,j,4}^{(1)} \\
	w_{i,j,1}^{(2)} & w_{i,j,2}^{(2)} & w_{i,j,3}^{(2)} & w_{i,j,4}^{(2)} \\
	\end{bmatrix},
\]
\[
	N_i = \begin{bmatrix}
	v_{i,1}^{(1)} & v_{i,2}^{(1)} & \ldots & v_{i,8}^{(1)} \\
	\end{bmatrix} \text{and }
%\]
%\[
	P = \begin{bmatrix}
	u_1^{(1)} & \ldots & u_{16}^{(1)} \\
	\end{bmatrix}
\]

\end{example}

\section{Properties of MRC with Hierarchical Locality}

% CAUTION: Fix notation below
In this section, we will derive two properties of MRC with hierarchical locality. We will show that the middle codes of a HDL/HL-MRC have to be data-local and local MRC respectively. Also, we derive the minimum distance of HDL MRC.

\begin{lem} \label{lem:midHDLMRC}
Consider a $[k, r_1, r_2, h_1, h_2, \delta]$ HDL-MRC $\mathcal{C}$. Let $A_i, 1 \leq i \leq t_1$ be the supports of the middle codes as defined in Definition \ref{defn:HDLC}. Then, for each $i$, $\mathcal{C}_{A_i}$ is a $[r_1, r_2, h_2, \delta]$ data-local MRC.
\end{lem}
\begin{proof}
Suppose not. This means that for some $i$, the middle code $\mathcal{C}_{A_i}$ is not a $[r_1, r_2, h_2, \delta]$ data-local MRC. By the definition of data-local MRC, we have that there exists a set $E_1 \subset A_i$ such that $|E_1| = r_1 + h_2$ and $\mathcal{C}_{E_1}$ is not an $[r_1+h_2, r_1, h_2+1]$ MDS code. This implies that there exists a subset $E' \subset E_1$ such that $|E'| = r_1$ and $\text{rank}(G|_{E'}) < r_1$. We can extend the set $E'$ to obtain a set $E \subset [n]$, $|E| = k+h_1$ which satisfies the conditions in the definition of HDL-MRC. The resulting punctured code $\mathcal{C}_{E}$ cannot be MDS since there exists an $r_1 < k$ sized subset of $E$ such that $\text{rank}(G|_{E'}) < r_1$.
\end{proof}

\begin{lem} \label{lem:midHLMRC}
Consider a $[k, r_1, r_2, h_1, h_2, \delta]$ HL-MRC $\mathcal{C}$. Let $A_i, 1 \leq i \leq t_1$ be the supports of the middle codes as defined in Definition \ref{defn:HLC}. Then, for each $i$, $\mathcal{C}_{A_i}$ is a $[r_1, r_2, h_2, \delta]$ local MRC.
\end{lem}
\begin{proof}
Proof is similar to the proof of Lemma \ref{lem:midHDLMRC}.
\end{proof}

\subsection{Minimum Distance of HDL-MRC}

\begin{lem}
The minimum distance of a $[k, r_1, r_2, h_1, h_2, \delta]$ HDL-MRC is given by $d = h_1 + h_2 + \delta +1$.
\end{lem}
\begin{proof}
Based on the definition of HDL-MRC, it can be seen that the $[k, r_1, r_2, h_1, h_2, \delta]$ HDL-MRC is a code with hierarchical locality as per Definition \ref{defn:hier_local} 
with $k, r_1,r_2$ being the same, $ \delta_2 -1 = \delta$, $\delta_1  =  h_2 + \delta +1$ and $n  =  k + h_1 + \frac{k}{r_1}(h_2 + \frac{r_1}{r_2}\delta)$
%with the following parameters:
%\begin{itemize}
%\item $k, r_1,r_2$ are the same.
%\item $ \delta_2 -1 = \delta$, $\delta_1  =  h_2 + \delta +1$.
%\item $n  =  k + h_1 + \frac{k}{r_1}(h_2 + \frac{r_1}{r_2}\delta)$.
%\end{itemize}
Substituting these parameters in the minimum distance bound in \eqref{eqn:min_dist_bound}, we have that $d \leq h_1 + h_2 + \delta +1$.

By Lemma \ref{lem:midHDLMRC}, we know that $\mathcal{C}_{A_i}$ is a $[r_1, r_2, h_2, \delta]$ data-local MRC. The minimum distance of $\mathcal{C}_{A_i}$ (from \eqref{eq:dminLMRC}) is $h_2 + \delta +1$. Thus, the middle code itself can recover from any $h_2 + \delta$ erasures. The additional $h_1$ erasures can be shown to be extended to a set $E$ (consisting of $k$ additional non-erased symbols) which satisfies the conditions in Definition \ref{defn:HDLMRC}. Since, the punctured code $\mathcal{C}|_E$ is a $[k+h_1,k,h_1+1]$ MDS code, it can be used to recover the $h_1$ erasures. Hence, $[k, r_1, r_2, h_1, h_2, \delta]$ HDL-MRC can recover from any $h_1 + h_2 + \delta$ erasures.
\end{proof}
 
% Note that $B'_{i, j} \subset A_i, \forall 1 \leq j \leq u$. We can say that, without any loss of generality that puncturing $\delta$ symbols from $B_{i,j}$ is equivalent to removing the $\delta$ parities we added. This means that, puncturing $B_{i,j}, \forall j$ leaves us with exactly $A_i$ which by our own definition is an $[r_1+h_2, r_1]$ code. Hence $A_i \forall i$ is an $[r_1, r_2, h_2]$ MRC as defined in \cite{mrc}

\subsection{Deriving HDL-MRC from HL-MRC}
% FIXME: For the new construction
In this section, we give a method to derive any HDL-MRC from a HL-MRC. Assume an $ [k, r_1, r_2, h_1, h_2, \delta] $ HL-MRC $\mathcal{C}$. Consider a particular set $E$ of $k+h_1$ symbols satisfying the conditions given in Definition \ref{defn:HLMRC}. We will refer to the elements of set $E$ as ``primary symbols". By the definition of HL-MRC, the code $\mathcal{C}$ when punctured to $E$ results in a $[k+h_1, k, h_1+1]$ MDS code. Hence, any $k$ subset of $E$ forms an information set. We will refer to the first $k$ symbols of $E$ as ``data symbols" and the rest $h_1$ symbols as global parities.
The symbols in $[n]\setminus E$ will be referred to as parity symbols (mid-level parities and local parities) and it can be observed that the parity symbols can be obtained as linear combinations of data symbols.
\begin{itemize}
    \item If $r_1 \mid h_1$ and $r_2 \mid h_2$,
    \begin{enumerate}
        \item For $A_i, \frac{k}{r_1} < i \leq \frac{k+h_1}{r_1}$, drop all the parity symbols, including $h_2$ mid-level parities per $A_i$ as well as the $\delta$ local parities per $B_{i,s} \subset A_i$. As a result, we would be left with $h_1$ ``primary symbols" in the local groups $A_i, \frac{k}{r_1} < i \leq \frac{k+h_1}{r_1}$. These form the global parities of the HDL-MRC. This step ensures that mid-level and local parities formed from global parities are dropped.
        
                \item For each $B_{i, s},\: 1 \leq i \leq \frac{k}{r_1},\: s > \frac{r_1}{r_2}$, drop the $\delta$ local parities. This step ensures that local parities formed from mid-level parities are dropped.

    \end{enumerate}
    This results in an $ [k, r_1, r_2, h_1, h_2, \delta] $ HDL-MRC.
    
    \item If $r_1 \nmid h_1$ and $r_2 \mid h_2$,
    \begin{enumerate}
        \item From the groups $A_i, \lfloor\frac{k}{r_1}\rfloor + 1 < i \leq \frac{k+h_1}{r_1}$, drop all the parity symbols, including $h_2$ mid-level parities per $A_i$ as well as the $\delta$ local parities per $B_{i,s} \subset A_i$.
        \item For each $B_{i, s},\: 1 \leq i \leq \lfloor\frac{k}{r_1}\rfloor,\: s > \frac{r_1}{r_2}$, drop the $\delta$ local parities.

        \item Drop the $k - \lfloor\frac{k}{r_1}\rfloor r_1$ data symbols in $A_i, i = \lfloor\frac{k}{r_1}\rfloor + 1$ and recalculate all the parities (local, mid-level and global) by setting these data symbols as zero in the linear combinations.
    \end{enumerate}
    This results in an $ [\lfloor\frac{k}{r_1}\rfloor r_1, r_1, r_2, h_1, h_2, \delta] $ HDL-MRC.  
\end{itemize}
For the case of $r_2 \nmid h_2$, HDL-MRC can be derived from HL-MRC using similar techniques as above. Hence, in the rest of the paper, we will discuss the construction of HL-MRC.

\section{General Construction of HL-MRC}

In this section, we will present a general construction of $ [k, r_1, r_2, h_1, h_2, \delta]$ HL-MRC. First, we will provide the structure of the code and then derive necessary and sufficient conditions for the code to be HL-MRC. Finally, we will apply a known result of BCH codes to complete the construction.

\begin{defn}
	A multiset $S \subseteq \mathbb{F}$ is $k$-wise independent over $\mathbb{F}$ if for every set $T \subseteq S$ such that $|T| \leq k$, $T$ is linearly independent over $\mathbb{F}$.
\end{defn}
\begin{lem} \label{lem:mds_linearized}
% This is taken from the Finite Fields by Lidl book.
Let $\mathbb{F}_{q^t}$ be an extension of $\mathbb{F}_q$. Let $a_1, a_2, \ldots, a_n$ be elements of $\mathbb{F}_{q^t}$. The following matrix

\[
\begin{bmatrix}
a_1 & a_2 & a_3 & \ldots & a_n \\
a_1^{q} & a_2^{q} & a_3^{q} & \ldots & a_n^{q} \\
\vdots & \vdots & \vdots & \ldots & \vdots \\
a_1^{q^{k-1}} & a_2^{q^{k-1}} & a_3^{q^{k-1}} & \ldots & a_n^{q^{k-1}} \\
\end{bmatrix}
\]
is the generator matrix of a $[n,k]$ MDS code if and only if $a_1, a_2, \ldots, a_n$ are $k$-wise linearly independent over $\mathbb{F}_q$.
\end{lem}
\begin{proof}
Directly follows from Lemma 3 in \cite{small}.
\end{proof}

\begin{constr} \label{constr:HLMRC}
The structure of the parity check matrix($H$) of a $ [k, r_1, r_2, h_1, h_2, \delta]$ HL-MRC is given by
\[
H = \begin{bmatrix}
    H_0   & \\
          & H_0 & \\
          &     & \ddots & \\
          &     &        &  H_0 \\
    H_1  & H_2  & \hdots &  H_{t_1} \\
    \end{bmatrix}
H_0 = \begin{bmatrix}
	M_0   & \\
		  & M_0 & \\
		  &     & \ddots & \\
		  &     &        &  M_0 \\
M_1  & M_2  & \hdots &  M_{t_2} \\
\end{bmatrix}
\]
Here, $H_0$ is an $(t_2 \delta  + h_2)\times n_1 $ matrix and $H_i, 1\leq i \leq t_1$ are an $h_1\times n_1$ matrix. $H_0$ is then further subdivided into $M_i$.
$M_0$ has the dimensions $\delta \times n_2$ and $M_i, 1\leq i \leq t_2$ is an $h_2 \times n_2$ matrix.

Assume $q$ to be a prime power such that $q \geq n$, $\mathbb{F}_{q^{m_1}}$ be an extension field of $\mathbb{F}_q$ and $\mathbb{F}_{q^m}$ is an extension field of $\mathbb{F}_{q^{m_1}}$, where $m_1 \mid m$. 
%Let $\beta \in \mathbb{F}_q$ be a primitive element. 
%Let $\Gamma = \{ \alpha_{i, s}, \alpha_{1, 2}, \ldots,  \alpha_{u, \delta+2} \} \in (\mathbb{F}_{q^M})^{u \cdot \delta+2}$ be a sequence of $u \cdot (\delta + 2)$ distinct elements. 

In this case, the construction is given by the following.
\[
M_0 = \begin{bmatrix}
        1 & 1 & 1 & \ldots & 1 \\
        0 & \beta & \beta^2 & \ldots & \beta^{n_2 -1} \\
        0 & \beta^2 & \beta^4 & \ldots & \beta^{2(n_2-1)} \\
        \vdots & \vdots & \vdots & \ldots & \vdots \\
        0 & \beta^{\delta-1} & \beta^{2(\delta - 1)} & \ldots & \beta^{(\delta - 1)(n_2-1)}
        \end{bmatrix},
\]
where $\beta \in \mathbb{F}_q$ is a primitive element.
\[
M_i = \begin{bmatrix}
        \alpha_{i, 1} & \alpha_{i, 2} & \ldots & \alpha_{i, n_2} \\
        \alpha_{i, 1}^{q} & \alpha_{i, 2}^{q} & \ldots & \alpha_{i, n_2}^{q} \\
        \vdots & \vdots & \ldots & \vdots \\
        \alpha_{i, 1}^{q^{h_2-1}} &  \alpha_{i, 2}^{q^{h_2-1}} & \ldots &  \alpha_{i, n_2}^{q^{h_2-1}} \\
        \end{bmatrix},
\]
where $i \in [t_2]$, $ \alpha_{i,j} \in \mathbb{F}_{q^{m_1}}, 1 \leq i \leq t_2, 1 \leq j \leq n_2$.
\[
H_i = [H_{i,1} \ H_{i,2} \ldots H_{i,t_2}]
\]
\[
H_{i,s} = \begin{bmatrix}
        \lambda_{i,s,1} & \lambda_{i,s, 2} & \ldots & \lambda_{i, s,n_2} \\
        \lambda_{i, s,1}^{q^{m_1}} & \lambda_{i, s,2}^{q^{m_1}} & \ldots & \lambda_{i, s,n_2}^{q^{m_1}} \\
        \vdots & \vdots & \ldots & \vdots \\
        \lambda_{i, s,1}^{q^{m_1(h_1-1)}} &  \lambda_{i, s,2}^{q^{m_1(h_1-1)}} & \ldots &  \lambda_{i, s,n_2}^{q^{m_1(h_1-1)}} \\
        \end{bmatrix},
\]

where $i \in [t_1], s \in [t_2]$, $ \lambda_{i,s,j} \in \mathbb{F}_{q^{m}}, 1 \leq i \leq t_1, 1 \leq s \leq t_2, 1 \leq j \leq n_2$.

%\[
%H_i = \begin{bmatrix}
%        v_{1, 1}^{(h_2+1)} & v_{1, 2}^{(h_2+1)} & \ldots & v_{1, r_2+\delta}^{(h_2+1)} & v_{2, 1}^{(h_2+1)} & \ldots & v_{u, r_2+\delta}^{(h_2+1)}\\
%        \end{bmatrix}
%\]        

\end{constr}

A $(\delta,h_2)$ erasure pattern is defined by the following two sets: 
%\begin{itemize}

$\Delta$ is a three dimensional array of indices with the first dimension $i$ indexing the middle code and hence $1 \leq i \leq t_1$, the second dimension $s$ indexing the local code and hence $1 \leq s \leq t_2$. The third dimension $j$ varies from $1$ to $\delta$ and used to index the $\delta$ coordinates which are erased in the $(i,s)^{\text{th}}$ group. Let $e \in [n]$ denote the actual index of the erased coordinate in the code and $e \in B_{i,s}$, then we set $\Delta_{i,s,j} = (e \mod n_2) + 1$. $\Delta_{i,s}$ is used to denote the vector of $\delta$ coordinates which are erased in the $(i,s)^{\text{th}}$ group. $\bar{\Delta}_{i,s}$ is used to denote the complement of $\Delta_{i,s}$ in the set $[n_2]$.

$\Gamma$ is a two dimensional array of indices with the first dimension $i$ indexing the middle code  and hence $1 \leq i \leq t_1$. The second dimension $j$ varies from $1$ to $h_2$ and used to index the additional $h_2$ coordinates which are erased in the $i^{\text{th}}$ group. Let $e \in [n]$ denote the actual index of the erased coordinate in the code and $e \in A_i$, then we set $\Gamma_{i,j} = (e \mod n_1) + 1$. $\Gamma_{i}$ is used to denote the vector of $h_2$ coordinates which are erased in the $i^{\text{th}}$ group. $\bar{\Gamma}_{i}$ is used to denote the complement of $\Gamma_{i}$ in the set $[n_1] \setminus (\cup_{s=1}^{t_2} \Delta_{i,s})$.
%\end{itemize}

We define some matrices and sets based on the parameters of the construction, which will be useful in proving the subsequent necessary and sufficient condition for the construction to be HL-MRC. Here, $\alpha_{s,\Delta_{i,s}}$ denotes the set $\{\alpha_{s,j} \mid j \in \Delta_{i,s}\}$.
\begin{eqnarray*}
L_{i,s} & = & (M_0|_{\Delta_{i,s}})^{-1}  M_0|_{\bar{\Delta}_{i,s}} \\
\Psi_i  & = & \{ \alpha_{s,\bar{\Delta}_{i,s}} + \alpha_{s,\Delta_{i,s}} L_{i,s}, 1 \leq s \leq t_2 \} \\
& = & \{ \Psi_{i,\Gamma_i}, \  \Psi_{i,\bar{\Gamma}_i} \} \\
& = & \{ \psi_{i,1}, \ldots, \psi_{i,h_2}, \psi_{i,h_2+1}, \ldots, \psi_{i, r_1+h_2} \}
\end{eqnarray*}
The above equalities follow by noting that the $\cup_{s=1}^{t_2} \bar{\Delta}_{i,s} = \Gamma_i \cup \bar{\Gamma}_i$.
We will refer to the elements in $\Psi_{i,\Gamma_i}$ by $ \{ \psi_{i,1}, \ldots, \psi_{i,h_2}\}$ and those in $\Psi_{i,\bar{\Gamma}_i}$ by $\{ \psi_{i,h_2+1}, \ldots, \psi_{i, r_1+h_2}\}$.
Consider the following matrix based on the elements of $\Psi_i$,
\begin{equation}
F_i = [F_i|_{\Gamma_i} \ \ F_i|_{\bar{\Gamma}_i}] = \begin{bmatrix}
        \psi_{i, 1} & \psi_{i, 2} & \ldots & \psi_{i, r_1+h_2} \\
        \psi_{i, 1}^{q} & \psi_{i, 2}^{q} & \ldots & \psi_{i, r_1+h_2}^{q} \\
        \vdots & \vdots & \ldots & \vdots \\
        \psi_{i, 1}^{q^{h_2-1}} &  \psi_{i, 2}^{q^{h_2-1}} & \ldots &  \psi_{i, r_1+h_2}^{q^{h_2-1}} \\
        \end{bmatrix},
\end{equation}
And
\begin{eqnarray*}
\Phi_i  & = & \{ \lambda_{i,s,\bar{\Delta}_{i,s}} + \lambda_{i,s,\Delta_{i,s}} L_{i,s}, 1 \leq s \leq t_2 \} \\
& = & \{ \Phi_{i,\Gamma_i}, \  \Phi_{i,\bar{\Gamma}_i} \} \\
& = & \{ \phi_{i,1}, \ldots, \phi_{i,h_2}, \phi_{i,h_2+1}, \ldots, \phi_{i, r_1+h_2} \}
\end{eqnarray*}

%\begin{equation}
Let $Z_i = (F_i|_{\Gamma_i})^{-1} F_i|_{\bar{\Gamma}_i}$.
%\end{equation}
Finally, the set $\Theta = \{ \Phi_{i,\bar{\Gamma}_i} + \Phi_{i,\Gamma_i} Z_i, 1 \leq i \leq t_1 \}$.

\begin{thm} \label{thm:necsuf}
The code described in Construction \ref{constr:HLMRC} is a $[k, r_1, r_2, h_1, h_2, \delta]$ HL-MRC only if, for any $(\delta,h_2)$ erasure pattern, %the following two conditions are satisfied:
each $\Psi_i, 1 \leq i \leq t_1$ is $h_2$-wise independent over $\mathbb{F}_q$ and
$\Theta$ is $h_1$-wise independent over $\mathbb{F}_{q^{m_1}}$.
%\begin{enumerate}
%\item Each $\Psi_i, 1 \leq i \leq t_1$ is $h_2$-wise independent over $\mathbb{F}_q$.
%\item $\Theta$ is $h_1$-wise independent over $\mathbb{F}_{q^{m_1}}$. 
%\end{enumerate}
\end{thm}

\begin{proof}
By Lemma \ref{lem:midHLMRC}, we have that $\mathcal{C}$ is a HL-MRC only if the $\mathcal{C}|_{A_i}$ is a $[r_1, r_2, h_2, \delta]$ local MRC.
By the definition of local MRC, a code is a $[r_1, r_2, h_2, \delta]$ local MRC, if after puncturing $\delta$ coordinates in each of the $\frac{r_1+h_2}{r_2}$ local groups, the resultant code is $[r_1+h_2, r_1, h_2+1]$ MDS code. 

The puncturing on a set of coordinates in the code is equivalent to shortening on the same set of coordinates in the dual code. Shortening on a set of coordinates in the dual code can be performed by zeroing the corresponding coordinates in the parity check matrix by row reduction. To prove that  $\mathcal{C}|_{A_i}$ is a $[r_1, r_2, h_2, \delta]$ local MRC, we need to show that certain punctured codes are MDS (Definition \ref{defn:local}). We will equivalently that the shortened codes of the dual code are MDS. 

Consider the coordinates corresponding to $(i,s)^\text{th}$ group in the parity check matrix. The sub-matrix of interest in this case is the following:
\begin{equation*}
\left [\begin{array}{c|c}
M_0|_{\Delta_{i,s}} & M_0|_{\bar{\Delta}_{i,s}} \\
\hline
\alpha_{s,\Delta_{i,s}} & \alpha_{s,\bar{\Delta}_{i,s}} \\
\alpha_{s,\Delta_{i,s}}^q & \alpha_{s,\bar{\Delta}_{i,s}}^q \\
\vdots & \vdots \\
\alpha_{s,\Delta_{i,s}}^{q^{h_2-1}} & \alpha_{s,\bar{\Delta}_{i,s}}^{q^{h_2-1}}
\end{array} \right],
\end{equation*} 
Where $\alpha_{s,\Delta_{i,s}}^q$ is the vector obtained by taking $q^{\text{th}}$ power of each element in the vector.
Applying row reduction to the above matrix, we have
\begin{equation*}
\left [\begin{array}{c|c}
M_0|_{\Delta_{i,s}} & M_0|_{\bar{\Delta}_{i,s}} \\
\hline
\bold{0} & \alpha_{s,\bar{\Delta}_{i,s}} + \alpha_{s,\Delta_{i,s}} L_{i,s}\\
 \bold{0} & (\alpha_{s,\bar{\Delta}_{i,s}} + \alpha_{s,\Delta_{i,s}} L_{i,s})^q \\
\vdots & \vdots \\
\bold{0} & (\alpha_{s,\bar{\Delta}_{i,s}} + \alpha_{s,\Delta_{i,s}} L_{i,s})^{q^{h_2-1}}
\end{array} \right].
\end{equation*} 
Note that $L_{i,s}$ can be pushed into the power of $q$ since the elements of $L_{i,s}$ are in $\mathbb{F}_q$. After row reducing $\delta$ coordinates from each of the $\frac{r_1+h_2}{r_2}$ local groups in $A_i$, the resultant parity check matrix is $F_i$. Applying Lemma \ref{lem:mds_linearized}, $F_i$ forms the generator matrix of an MDS code if and only if the set $\Psi_i$ is $h_2$-wise independent over $\mathbb{F}_q$.
The shortening of the code above is applicable to mid-level parities. Now, we will apply similar shortening in two steps to global parities. The sub-matrix of interest in this case is the following:
\begin{equation*}
\left [\begin{array}{c|c}
M_0|_{\Delta_{i,s}} & M_0|_{\bar{\Delta}_{i,s}} \\
\hline
\alpha_{s,\Delta_{i,s}} & \alpha_{s,\bar{\Delta}_{i,s}} \\
\alpha_{s,\Delta_{i,s}}^q & \alpha_{s,\bar{\Delta}_{i,s}}^q \\
\vdots & \vdots \\
\alpha_{s,\Delta_{i,s}}^{q^{h_2-1}} & \alpha_{s,\bar{\Delta}_{i,s}}^{q^{h_2-1}} \\
\hline
\lambda_{i,s,\Delta_{i,s}} & \lambda_{i,s,\bar{\Delta}_{i,s}} \\
\lambda_{i,s,\Delta_{i,s}}^{q^{m_1}} & \lambda_{i,s,\bar{\Delta}_{i,s}}^{q^{m_1}} \\
\vdots & \vdots \\
\lambda_{i,s,\Delta_{i,s}}^{q^{m_1(h_1-1)}} & \lambda_{i,s,\bar{\Delta}_{i,s}}^{q^{m_1(h_1-1)}}
\end{array} \right]
\end{equation*} 
Applying row reduction to the above matrix, we have
\begin{equation*}
\left [\begin{array}{c|c}
M_0|_{\Delta_{i,s}} & M_0|_{\bar{\Delta}_{i,s}} \\
\hline
\bold{0} & \alpha_{s,\bar{\Delta}_{i,s}} + \alpha_{s,\Delta_{i,s}} L_{i,s}\\
 \bold{0} & (\alpha_{s,\bar{\Delta}_{i,s}} + \alpha_{s,\Delta_{i,s}} L_{i,s})^q \\
\vdots & \vdots \\
\bold{0} & (\alpha_{s,\bar{\Delta}_{i,s}} + \alpha_{s,\Delta_{i,s}} L_{i,s})^{q^{h_2-1}} \\
\hline
\bold{0} & \lambda_{i,s,\bar{\Delta}_{i,s}} + \lambda_{i,s,\Delta_{i,s}} L_{i,s}\\
 \bold{0} & (\lambda_{i,s,\bar{\Delta}_{i,s}} + \lambda_{i,s,\Delta_{i,s}} L_{i,s})^{q^{m_1}} \\
\vdots & \vdots \\
\bold{0} & (\lambda_{i,s,\bar{\Delta}_{i,s}} + \lambda_{i,s,\Delta_{i,s}} L_{i,s})^{q^{m_1(h_1-1)}}
\end{array} \right].
\end{equation*} 
To apply row reduction again, we consider the following submatrix obtained by deleting the zero columns and aggregating the non-zero columns from the $\frac{r_1+h_2}{r_2}$ groups, 
\begin{equation*}
\left [\begin{array}{c|c}
F_i|_{\Gamma_i} & F_i|_{\bar{\Gamma}_i} \\
\hline
\Phi_{i,\Gamma_i} & \Phi_{i,\bar{\Gamma}_i} \\
\Phi_{i,\Gamma_i}^{q^{m_1}} & \Phi_{i,\bar{\Gamma}_i}^{q^{m_1}} \\
\vdots & \vdots \\
\Phi_{i,\Gamma_i}^{q^{m_1(h_1-1)}} & \Phi_{i,\bar{\Gamma}_i}^{q^{m_1(h_1-1)}}
\end{array} \right].
\end{equation*}
Applying row reduction to the above matrix, we have
\begin{equation*}
\left [\begin{array}{c|c}
F_i|_{\Gamma_i} & F_i|_{\bar{\Gamma}_i} \\
\hline
\bold{0} & \Phi_{i,\bar{\Gamma}_i} + \Phi_{i,\Gamma_i} Z_i \\
\bold{0} & (\Phi_{i,\bar{\Gamma}_i} + \Phi_{i,\Gamma_i} Z_i)^{q^{m_1}} \\
\vdots & \vdots \\
\bold{0} & (\Phi_{i,\bar{\Gamma}_i} + \Phi_{i,\Gamma_i} Z_i)^{q^{m_1(h_1-1)}}
\end{array} \right].
\end{equation*}
Note that $Z_i$ can be pushed into the power of $q^{m_1}$ since the elements of $Z_i$ are in $\mathbb{F}_{q^{m_1}}$. Applying Lemma \ref{lem:mds_linearized}, the row reduced matrix above forms the generator matrix of an MDS code if and only if the set $\Theta$ is $h_1$-wise independent over $\mathbb{F}_{q^{m_1}}$.
 \end{proof}
 
 \begin{lem} \label{lem:hwise_ind}
 For any $(\delta,h_2)$ erasure pattern,
 \begin{itemize}
 \item For each $i$, $\Psi_i = \{ \alpha_{s,\bar{\Delta}_{i,s}} + \alpha_{s,\Delta_{i,s}} L_{i,s}, 1 \leq s \leq t_2 \}$ is $h_2$-wise independent over $\mathbb{F}_q$ if the set $\{ \alpha_{s,j}, 1 \leq s \leq t_2, 1 \leq j \leq n_2 \}$
% \begin{equation*}
% \{ \alpha_{s,j}, 1 \leq s \leq t_2, 1 \leq j \leq n_2 \}
% \end{equation*}
 is $(\delta+1) h_2$-wise independent over $\mathbb{F}_q$.
 \item $\Theta = \{ \Phi_{i,\bar{\Gamma}_i} + \Phi_{i,\Gamma_i} Z_i, 1 \leq i \leq t_1 \}$ is $h_1$-wise independent over $\mathbb{F}_{q^{m_1}}$ if the set $\{ \lambda_{i,s,j}, 1 \leq i \leq t_1, 1 \leq s \leq t_2, 1 \leq j \leq n_2 \}$
% \begin{equation*}
% \{ \lambda_{i,s,j}, 1 \leq i \leq t_1, 1 \leq s \leq t_2, 1 \leq j \leq n_2 \}
% \end{equation*}
 is $(\delta+1)(h_2+1)h_1$-wise independent over $\mathbb{F}_{q^{m_1}}$.
 \end{itemize}
 
 \end{lem}
 \begin{proof}
 Since the size of matrix $L_{i,s}$ is $\delta \times (n_2 - \delta)$, each element of $\Psi_i$ can be a $\mathbb{F}_q$-linear combination of atmost $\delta + 1$ different $\alpha_{s,j}$. Consider $\mathbb{F}_q$-linear combination of $h_2$ elements in $\Psi_i$.  The linear combination will have at most $(\delta + 1) h_2$ different $\alpha_{s,j}$. Thus, if the set $\{ \alpha_{s,j} \}$ is $(\delta+1) h_2$-wise independent over $\mathbb{F}_q$, then $\Psi_i$ is $h_2$-wise independent over $\mathbb{F}_q$. To prove the second part, we note that each element of $\Phi_i$ is a linear combination of at most $\delta + 1$ different $\lambda_{i,s,j}$. Since the size of the matrix $Z_i$ is $h_2 \times r_1$, each element of $\Theta$ can be a $\mathbb{F}_{q^{m_1}}$-linear combination of atmost $(\delta+1)(h_2+1)$ different $\lambda_{i,s,j}$. Consider $\mathbb{F}_{q^{m_1}}$-linear combination of $h_1$ elements in $\Theta$.  The linear combination will have at most $(\delta + 1) (h_2+1)h_1$ different $\lambda_{i,s,j}$. Thus, if the set $\{ \lambda_{i,s,j} \}$ is $(\delta + 1) (h_2+1)h_1$-wise independent over $\mathbb{F}_{q^{m_1}}$, then $\Theta$ is $h_1$-wise independent over $\mathbb{F}_{q^{m_1}}$.
 \end{proof}
 
 We will design the $\{ \alpha_{s,j} \}$ and $\{ \lambda_{i,s,j} \}$ based on the Lemma \ref{lem:hwise_ind} so that the field size is minimum possible. We will pick these based on the following two properties:
 \begin{itemize}
 \item {\bf Property 1:} The columns of parity check matrix of an $[n,k,d]$ linear code over $\mathbb{F}_q$ can be interpreted as $n$ elements over $\mathbb{F}_{q^{n-k}}$ which are $(d-1)$-wise linear independent over $\mathbb{F}_q$.
 \item {\bf Property 2:} There exists $[n = q^t-1, k,d]$ BCH codes over $\mathbb{F}_q$ \cite{roth2006}, where the parameters are related as
 \begin{equation*}
 n-k  = 1 + \left \lceil \frac{q-1}{q} (d-2) \right \rceil  \lceil \log_2 (n)\rceil .
 \end{equation*}
 \end{itemize}
 
 \begin{thm} \label{thm:HLBCH}
 The code in Construction \ref{constr:HLMRC} is a $ [k, r_1, r_2, h_1, h_2, \delta]$ HL-MRC if the parameters are picked as follows:
 \begin{enumerate}
\item $q$ is the smallest prime power greater than $n_2$.
\item $m_1$ is chosen based on the following relation:
%\begin{equation*}
$m_1 = 1 + \left \lceil \frac{q-1}{q} ((\delta+1)h_2-1) \right \rceil  \lceil \log_q (n_2 t_2)\rceil$.
%\end{equation*}
\item $n_2 t_2$ elements $\{ \alpha_{s,j} \}$ over $\mathbb{F}_{q^{m_1}}$ are set to be the columns of parity check matrix of the BCH code over $\mathbb{F}_q$ with parameters $[n = q^{\lceil \log_q (n_2 t_2) \rceil}-1, q^{\lceil \log_q (n_2 t_2)\rceil}-1-m_1, (\delta+1)h_2+1]$ .
\item $m$ is chosen to be the smallest integer dividing $m_1$ based on the following relation:
%\begin{equation*}
$m \geq 1 + \left \lceil \frac{q^{m_1}-1}{q^{m_1}} ((\delta+1)(h_2+1)h_1-1) \right \rceil \lceil \log_{q^{m_1}} (n)\rceil$.
%\end{equation*}
\item $n$ elements $\{ \lambda_{i,s,j} \}$ over $\mathbb{F}_{q^{m}}$ are set to be the columns of parity check matrix of the BCH code over $\mathbb{F}_{q^{m_1}}$ with parameters $[n = q^{m_1\lceil \log_{q^{m_1}} (n) \rceil}-1, q^{m_1\lceil \log_{q^{m_1}} (n) \rceil}-1-m, (\delta+1)(h_2+1)h_1+1]$ .
\end{enumerate}
 \end{thm}
\begin{proof}
The proof follows from Lemma \ref{lem:hwise_ind} and Properties 1 and 2.
\end{proof}

%The values for $\alpha_{i,j}$ and $\beta$ are chosen so that they satisfy \textit{Lemma 4} from \cite{pmds}
%
%Let $M= M_0$, $u=\frac{r_1+h_2}{r_2}$ and $S = (s_1, \ldots, s_u) \in ([r_2+\delta]^\delta)^u$. 
%Also let $\bar{s}_i =(\bar{s}_{i,1}, \bar{s}_{i, 2}, \ldots, \bar{s}_{i, r_2})$ which contains, in increasing order all the elements in $[r_2+\delta]$ that do not appear in $s_i$
%
%Let $M_{s_i}$ be the sumbatrix formed by restricting $M$ to columns in $s_i$. Since $M$ is the parity check matrix for an $[r_2+\delta, r_2]$ code, any $\delta$ columns of $M$ are linearly independent. Hence $M_{s_i}$ is invertible. Let, \[A_{s_i} = (M_{s_i})^{-1} \cdot M_{\bar{s}_i}\]
%
%Let $\alpha_i$ be the vector $\alpha_i = (\alpha_{i,1}, \alpha_{i,2}, \ldots, \alpha_{i,r_2+\delta})$. Extending this, let $\alpha_{s_i}$ be the subvector of $\alpha_i$ at coordinates specified by $s_i$. Define the vector $\gamma_{s_i} = (\gamma_{s_i,1}, \gamma_{s_i,2}, \ldots, \gamma_{s_i,r_2})$ as \[\gamma_{s_i} = \alpha_{s_i} \cdot A_{s_i} + \alpha_{\bar{s}_i}\]
%
%Now,$H_0$ is an $(\delta;h_2)$ PMDS code if and only if for any $S = (s_1, s_2, \ldots, s_u) \in ([r_2+\delta]^\delta)^u$, the set $T = \{\gamma_{s_i,j}\}_{i\in[u], j\in[r_2]}$ is $h_2$-wise independent over $\mathbb{F}_q$.
%
\vspace{-2ex}
\section{HL-MRC Construction for $h_1 = 1$}
In this section, we present a construction of HL-MRC for the case when $h_1=1$ over a field size lower than that provided by Construction \ref{constr:HLMRC}.

\begin{constr}\label{constr:h1}
	The structure of the parity check matrix for the present construction is the same as that given in Construction \ref{constr:HLMRC}. In addition, the matrices $M_0$ and $M_i,\: 1 \leq i \leq t_2$ also remain the same. We modify the matrix $H_i, \; 1 \leq i \leq t_1$ as follows:
	\[ H_i = \begin{bmatrix}
	\alpha_{1, 1}^{q^{h_2}} & \alpha_{1, 2}^{q^{h_2}} & \ldots \alpha_{t_2, n_2}^{q^{h_2}}
	\end{bmatrix},
	\]
where $\{ \alpha_{s,j} \in \mathbb{F}_{q^{m_1}}, 1 \leq s \leq t_2, 1 \leq j \leq n_2 \}$ are chosen to be $(\delta+1)(h_2+1)$-wise independent over $\mathbb{F}_q$ based on Theorem \ref{thm:HLBCH}.

\end{constr}

\begin{thm}
The code $C$ given by Construction \ref{constr:h1} is a $ [k, r_1, r_2, h_1=1, h_2, \delta] $ HL-MRC.
\end{thm}

\begin{proof}
We show that $H$ can be used to correct all erasure patterns defined in Definition \ref{defn:HLMRC}. From the definition the code should recover from $\delta$ erasures per $B_{i,s}$, $h_2$ additional erasures per $A_i$ and $1$ more erasure anywhere in the entire code.

Now, with $h_1=1$, the last erasure can be part of one group. Thus, effectively the code should recover from $h_2+1$ erasures per group. Suppose that the last erasure is in the $i^{\text{th}}$ group. The submatrix of interest for the $(i,s)^{th}$ local group is
\begin{equation*}
\left [\begin{array}{c|c}
M_0|_{\Delta_{i,s}} & M_0|_{\bar{\Delta}_{i,s}} \\
\hline
\alpha_{s,\Delta_{i,s}} & \alpha_{s,\bar{\Delta}_{i,s}} \\
\alpha_{s,\Delta_{i,s}}^q & \alpha_{s,\bar{\Delta}_{i,s}}^q \\
\vdots & \vdots \\
\alpha_{s,\Delta_{i,s}}^{q^{h_2-1}} & \alpha_{s,\bar{\Delta}_{i,s}}^{q^{h_2-1}} \\
\hline
\alpha_{s,\Delta_{i,s}}^{q^{h_2}} & \alpha_{s,\bar{\Delta}_{i,s}}^{q^{h_2}} \\
\end{array} \right].
\end{equation*}

Following the proof of Theorem \ref{thm:necsuf} and performing row reduction of $\delta$ coordinates, the resultant matrix  is
\[
\begin{bmatrix}
	\psi_{i, 1} & \psi_{i, 2} & \ldots & \psi_{i, r_1+h_2} \\
	\psi_{i, 1}^{q} & \psi_{i, 2}^{q} & \ldots & \psi_{i, r_1+h_2}^{q} \\
	\vdots & \vdots & \ldots & \vdots \\
	\psi_{i, 1}^{q^{h_2-1}} &  \psi_{i, 2}^{q^{h_2-1}} & \ldots &  \psi_{i, r_1+h_2}^{q^{h_2-1}} \\
	\psi_{i, 1}^{q^{h_2}} &  \psi_{i, 2}^{q^{h_2}} & \ldots &  \psi_{i, r_1+h_2}^{q^{h_2}} \\
\end{bmatrix}.
\]	
Now, by Lemma \ref{lem:mds_linearized}, it is the generator matrix of an MDS code if and only if $\Psi_i$ is $(h_2+1)$-wise independent over $\mathbb{F}_q$.
\end{proof}
\vspace{-3ex}
\section*{Acknowledgment} This work was supported partly by the Early Career Research Award (ECR/2016/000954) from Science and Engineering Research Board (SERB) to V. Lalitha.

%\medskip
\bibliography{biblo}{}
\bibliographystyle{ieeetr}

\end{document}